\documentclass[12pt,reqno]{amsart}

\usepackage{mathrsfs,amsfonts,amsthm,amsmath,amssymb,thmtools}
\usepackage[most]{tcolorbox}

\usepackage{enumitem}                
\usepackage{lineno}                 

\usepackage[hyperindex, pdfencoding = auto, psdextra, bookmarksdepth = 4]{hyperref}  
\hypersetup{
    bookmarksopen      = true,
    bookmarksopenlevel = 1,
    bookmarksnumbered  = true,
    pdftitle    = {},
    pdfauthor   = {},
    linktoc     = page,
    anchorcolor = green,
    colorlinks  = true,        
    linkcolor   = black,
    citecolor   = blue,
    filecolor   = blue,
    urlcolor    = magenta,
}

\usepackage[nameinlink]{cleveref}    

\theoremstyle{plain}

\newtheorem{theorem}{Theorem}[section]
\newtheorem{lemma}[theorem]{Lemma}
\newtheorem{proposition}[theorem]{Proposition}

\theoremstyle{definition}   
\newtheorem{definition}[theorem]{Definition}
\newtheorem{remark}[theorem]{Remark}
\newtheorem{example}[theorem]{Example}

\newtheorem{innercustomthm}{Theorem}
\newenvironment{customthm}[1]
  {\renewcommand\theinnercustomthm{#1}\innercustomthm}
  {\endinnercustomthm}

\makeatletter
    \def\subsection{\@startsection{subsection}{2}%
    \z@{.5\linespacing\@plus.7\linespacing}{.3\linespacing}%
    {\normalfont\bfseries}}
\makeatother 

\makeatletter                                         
    \newcommand{\LeftEqNo}{\let\veqno\@@leqno}        
\makeatother

\numberwithin{equation}{section}                      

\begin{document}

\

\vspace{-2cm}

\title[Bargmann invariants and local unitary equivalence]{Bargmann invariants and local unitary equivalence}

\author{Minghui Ma}
\address{Minghui Ma, School of Mathematical Sciences, Dalian University of Technology, Dalian, 116024, China}
\email{minghuima@dlut.edu.cn}

\author{Rui Shi}
\address{Rui Shi, School of Mathematical Sciences, Dalian University of Technology, Dalian, 116024, China}
\email{ruishi@dlut.edu.cn}

\thanks{}


\begin{abstract}
In this paper, we study the local unitary equivalence of quantum states on $\mathbb{C}^n\otimes\mathbb{C}^n$, which is an important notion in quantum information theory.
For the case $n=2$, the local unitary orbits of two-qubit states are completely determined by their local unitary Bargmann invariants.
We show that the local unitary Bargmann invariants do not form a complete set of invariants of local unitary orbits for $n\geqslant 3$, which negatively answers a problem proposed by L. Zhang and B. Xie.
\end{abstract}

\maketitle

\section{Introduction}

In quantum information theory, quantum states provide a fundamental mathematical framework for characterizing the physical state of a quantum system, encoding the probabilities, correlations, and fluctuations of all accessible observables.
However, when evaluating intrinsic non-local properties, not all mathematical distinctions among quantum states correspond to physically distinct quantum resources.
In particular, two states that differ only by a local unitary transformation possess identical entanglement spectra and correlation structures.
This motivates the study of local unitary equivalence, defined by the action of the local unitary group \cite{JFLLZ-15,Kra-10,Mak-02,ZXT-25}.

In a landmark result, Y. Makhlin \cite{Mak-02} explicitly constructed a complete set of 18 local unitary invariants for two-qubit states acting on $\mathbb{C}^2\otimes\mathbb{C}^2$.
Recently, L. Zhang, B. Xie, and Y. Tao \cite{ZXT-25} established a formal connection between Makhlin's fundamental invariants and local unitary Bargmann invariants.
Motivated by this relationship, L. Zhang and B. Xie conjectured in \cite[Section 8.2]{ZX-26} that the local unitary orbit of an arbitrary multipartite state is completely determined by its Bargmann invariants.
This conjecture is conceptually analogous to Kaplansky's second problem \cite{Ber-87,DH-90,Jia-04,KS-57,Kap-54,MRTZ-24}, which asserts that the global equivalence of $A\oplus A$ and $B\oplus B$ implies the local equivalence of $A$ and $B$ for certain equivalence relations.
Moreover, we observe that local unitary equivalence and Kaplansky's second problem are closely related.
The main purpose of this paper is to disprove this conjecture by constructing explicit counterexamples for two-qudit states on $\mathbb{C}^n\otimes\mathbb{C}^n$ for $n\geqslant 3$.

The Choi-Jamiołkowski isomorphism \cite{Choi-75,Jam-72} is one of the most fundamental and widely utilized tools in quantum information theory.
Let $\{E_{ij}\}_{1\leqslant i,j\leqslant m}$ be a system of matrix units in $M_m(\mathbb{C})$.
For a linear map $\Phi\colon M_m(\mathbb{C})\to M_n(\mathbb{C})$, the corresponding \emph{Choi matrix} $\rho_\Phi$ acting on the bipartite tensor product space $\mathbb{C}^m\otimes\mathbb{C}^n$ is defined as
\begin{equation}\label{equ Choi}
	\rho_\Phi=\sum_{1\leqslant i,j\leqslant m}E_{ij}\otimes\Phi(E_{ij})\in M_m(\mathbb{C})\otimes M_n(\mathbb{C})\cong M_{mn}(\mathbb{C}).
\end{equation}
It is well-known that $\Phi$ is a completely positive map if and only if $\rho_\Phi$ is a positive semi-definite matrix.
The map $\Phi$ is called a \emph{quantum channel} if it is completely positive and trace-preserving.

Let us introduce basic definitions and notation.
Let $D(\mathbb{C}^n)$ be the set of all positive semi-definite matrices $\rho$ in $M_n(\mathbb{C})$ satisfying that $\mathrm{Tr}(\rho)=1$.
Every matrix in $D(\mathbb{C}^n)$ is called a \emph{quantum state} or a \emph{density matrix}.
For every matrix $\rho_{AB}$ in $M_m(\mathbb{C})\otimes M_n(\mathbb{C})$, its \emph{partial traces} are denoted by
\begin{equation*}
  \rho_A=(\mathrm{Id}\otimes\mathrm{Tr})(\rho_{AB}),\quad
  \rho_B=(\mathrm{Tr}\otimes\mathrm{Id})(\rho_{AB}).
\end{equation*}
It is clear that $\rho_A\in D(\mathbb{C}^m)$ and $\rho_B\in D(\mathbb{C}^n)$ for every $\rho_{AB}\in D(\mathbb{C}^m\otimes\mathbb{C}^n)$.
For a linear map $\Phi\colon M_m(\mathbb{C})\to M_n(\mathbb{C})$, let $\rho_{AB}\in M_m(\mathbb{C})\otimes M_n(\mathbb{C})$ be the Choi matrix given by \eqref{equ Choi}.
Then $\rho_A=\sum_{1\leqslant i,j\leqslant m}\mathrm{Tr}(\Phi(E_{ij}))E_{ij}$ and $\rho_B=\Phi(I_m)$.
It follows that
\begin{enumerate}
\item[$(1)$] $\rho_A=\alpha I_m$ if and only if $\mathrm{Tr}(\Phi(X))=\alpha\mathrm{Tr}(X)$ for every $X\in M_n(\mathbb{C})$, i.e., $\Phi$ is trace-preserving up to a scaling factor;

\item[$(2)$] $\rho_B=\beta I_n$ if and only if $\Phi(I_m)=\beta I_n$, i.e., $\Phi$ is unital up to a scaling factor.
\end{enumerate}
If both partial traces of $\rho_{AB}\in D(\mathbb{C}^m\otimes\mathbb{C}^n)$ are scalar multiples of their respective identity matrices, then we say that $\rho_{AB}$ is a \emph{maximally mixed marginal state}, and the corresponding $\Phi$ is unital and trace-preserving up to scaling factors.

Let $\Psi=(\rho_1,\rho_2,\ldots,\rho_K)$ and $\Psi'=(\rho'_1,\rho'_2,\ldots,\rho'_K)$ be two tuples of quantum states in $D(\mathbb{C}^n)$.
We say that $\Psi$ and $\Psi'$ are \emph{unitarily equivalent} if there exists a unitary matrix $U\in M_n(\mathbb{C})$ such that $\rho'_j=U\rho_jU^*$ for every $1\leqslant j\leqslant K$.
For every sequence $j_1,j_2,\ldots,j_m\in\{1,2,\ldots,K\}$, the corresponding \emph{Bargmann invariant} is defined as
\begin{equation*}
	\Delta_{j_1,j_2,\ldots,j_m}(\Psi)=\mathrm{Tr}(\rho_{j_1}\rho_{j_2}\cdots\rho_{j_m}).
\end{equation*}
It is clear that unitarily equivalent tuples have the same Bargmann invariants.
Conversely, if all their Bargmann invariants coincide, then they are unitarily equivalent by \cite{Pro-76}.

\begin{definition}
Let $\rho_{AB}$ and $\sigma_{AB}$ be two bipartite states in $D(\mathbb{C}^m\otimes\mathbb{C}^n)$.
We say that $\rho_{AB}$ and $\sigma_{AB}$ are \emph{locally unitarily equivalent} if there are unitary matrices $U\in M_m(\mathbb{C})$ and $V\in M_n(\mathbb{C})$ such that
\begin{equation*}
  \sigma_{AB}=(U\otimes V)\rho_{AB}(U\otimes V)^*.
\end{equation*}
The \emph{local unitary Bargmann invariants} of $\rho_{AB}\in D(\mathbb{C}^m\otimes\mathbb{C}^n)$ are defined as the Bargmann invariants of the triple $(\rho_{AB},\rho_A\otimes I_n,I_m\otimes\rho_B)$.
\end{definition}

By definition, locally unitarily equivalent bipartite states have the same local unitary Bargmann invariants.
Conversely, the authors in \cite{ZXT-25} proved that the local unitary equivalence is completely determined by the local unitary Bargmann invariants for two-qubit states in $D(\mathbb{C}^2\otimes\mathbb{C}^2)$.
By \cite{Pro-76}, $\rho_{AB}$ and $\sigma_{AB}$ have the same local unitary Bargmann invariants if and only if the triples $(\rho_{AB},\rho_A\otimes I_n,I_m\otimes\rho_B)$ and $(\sigma_{AB},\sigma_A\otimes I_n,I_m\otimes\sigma_B)$ are unitarily equivalent.
Therefore, the authors in \cite[Section 8.2]{ZX-26} proposed the following problem:
For bipartite states $\rho_{AB}$ and $\sigma_{AB}$ in $D(\mathbb{C}^m\otimes\mathbb{C}^n)$, if there exists a unitary matrix $W$ in $M_m(\mathbb{C})\otimes M_n(\mathbb{C})$ such that
\begin{equation*}
	\begin{cases}
		\sigma_{AB} & =W\rho_{AB}W^*,\\
		\sigma_A\otimes I_n &=W(\rho_A\otimes I_n)W^*,\\
		I_m\otimes\sigma_B &=W(I_m\otimes\rho_B)W^*,
	\end{cases}
\end{equation*}
does it hold that $\sigma_{AB}$ and $\rho_{AB}$ are locally unitarily equivalent?

If we only focus on the last two relations
\begin{equation*}
	\sigma_A\otimes I_n=W(\rho_A\otimes I_n)W^*\quad\text{and}\quad I_m\otimes\sigma_B=W(I_m\otimes\rho_B)W^*,
\end{equation*}
then this problem degenerates to Kaplansky's second test problem up to unitary equivalence.
It follows from \cite{KS-57} that there are unitary matrices $U\in M_m(\mathbb{C})$ and $V\in M_n(\mathbb{C})$ such that
\begin{equation*}
	\sigma_A=U\rho_AU^*\quad\text{and}\quad\sigma_B=V\rho_BV^*.
\end{equation*}
But we do not have $\sigma_{AB}=(U\otimes V)\rho_{AB}(U\otimes V)^*$ in general.
In fact, we solve the above problem \cite[Section 8.2]{ZX-26} negatively as follows.

\begin{customthm}{\ref{thm main}}
For each $n\geqslant 3$, there are two-qudit states $\rho_{AB}$ and $\sigma_{AB}$ in $D(\mathbb{C}^n\otimes\mathbb{C}^n)$ and a unitary matrix $W$ in $M_n(\mathbb{C})\otimes M_n(\mathbb{C})$ such that
\begin{equation*}
	\begin{cases}
		\sigma_{AB} & =W\rho_{AB}W^*,\\
		\sigma_A\otimes I_n &=W(\rho_A\otimes I_n)W^*,\\
		I_n\otimes\sigma_B &=W(I_n\otimes\rho_B)W^*,
	\end{cases}
\end{equation*}
while $\sigma_{AB}$ and $\rho_{AB}$ are not locally unitarily equivalent.
\end{customthm}

As a direct consequence, for each $n\geqslant 3$, the local unitary equivalence cannot be completely determined by the local unitary Bargmann invariants for two-qudit states in $D(\mathbb{C}^n\otimes\mathbb{C}^n)$.
On the other hand, by \Cref{prop rank-one}, the local unitary equivalence of pure states in $D(\mathbb{C}^m\otimes\mathbb{C}^n)$ can be determined by the local unitary Bargmann invariants for all $m,n\geqslant 1$.

\section{The proof of main result}

We begin with a technical lemma that will be used in the proof of \Cref{lem 2-sigma}.

\begin{lemma}\label{lem P+Q-1}
Let $Q_1,Q_2$ be rank-one projections in $M_n(\mathbb{C})$.
Then for every $c\ne 0$, $cQ_1+Q_2-I_n$ is invertible in $M_n(\mathbb{C})$ if and only if $Q_1Q_2\ne 0$.
\end{lemma}

\begin{proof}
If $Q_1Q_2=0$, then $(cQ_1+Q_2-I_n)Q_2=0$, and hence $cQ_1+Q_2-I_n$ is not invertible in $M_n(\mathbb{C})$.
Conversely, suppose that $cQ_1+Q_2-I_n$ is not invertible.
Then there exists a nonzero vector $x\in\mathbb{C}^n$ such that $(cQ_1+Q_2-I_n)x=0$, i.e., $cQ_1x=(I_n-Q_2)x$.
If $Q_1x=0$, then $(I_n-Q_2)x=0$, i.e., $x=Q_2x$.
It follows that $\mathrm{ran}(Q_2)=\mathbb{C}x$.
Thus, $Q_1Q_2=0$.
If $Q_1x\ne 0$, then
\begin{equation*}
	\mathrm{ran}(Q_1)=\mathbb{C}Q_1x\subseteq\mathrm{ran}(I_n-Q_2).
\end{equation*}
It follows that $Q_1\leqslant I_n-Q_2$, and hence $Q_1Q_2=0$.
This completes the proof.
\end{proof}

Next we try to find maximally mixed marginal states in $D(\mathbb{C}^m\otimes\mathbb{C}^n)$.
Suppose that $\rho_{AB}$ is a self-adjoint matrix in $M_m(\mathbb{C})\otimes M_n(\mathbb{C})$ such that $\rho_A=\alpha I_m$ and $\rho_B=\beta I_n$, where $\alpha=\frac{1}{m}\mathrm{Tr}(\rho_{AB})$ and $\beta=\frac{1}{n}\mathrm{Tr}(\rho_{AB})$.
Choose $c>\|\rho_{AB}\|$ and let
\begin{equation*}
	\rho'_{AB}=\frac{1}{\mathrm{Tr}(\rho_{AB}+cI_m\otimes I_n)}(\rho_{AB}+cI_m\otimes I_n).
\end{equation*}
Then $\rho'_{AB}$ is a maximally mixed marginal state in $D(\mathbb{C}^m\otimes\mathbb{C}^n)$.
Therefore, we only need to focus on self-adjoint matrices in $M_m(\mathbb{C})\otimes M_n(\mathbb{C})$.

Suppose that $P_1,P_2$ are rank-one projections in $M_m(\mathbb{C})$, $B_1,B_2,B_3$ are matrices in $M_n(\mathbb{C})$, and
\begin{equation*}
	\rho_{AB}=P_1\otimes B_1+P_2\otimes B_2+I_m\otimes B_3\in M_m(\mathbb{C})\otimes M_n(\mathbb{C}).
\end{equation*}
By the definition of partial traces, we have
\begin{equation*}
	\rho_A=\mathrm{Tr}(B_1)P_1+\mathrm{Tr}(B_2)P_2+\mathrm{Tr}(B_3)I_m.
\end{equation*}
If we assume that $\mathrm{Tr}(B_1)=\mathrm{Tr}(B_2)=0$, then $\rho_A$ is a scalar multiple of $I_m$.
Inspired by \Cref{lem P+Q-1}, we take $B_1=cQ_1+Q_2-I_n$, where $Q_1,Q_2$ are rank-one projections in $M_n(\mathbb{C})$.
By the assumption
\begin{equation*}
	\mathrm{Tr}(B_1)=c+1-n=0,
\end{equation*}
we have $c=n-1$.
In other words, $B_1=(n-1)Q_1+Q_2-I_n$.
Similarly, we take $B_2=Q_1+(n-1)Q_2-I_n$.
Then
\begin{equation*}
	\rho_B=B_1+B_2+mB_3=n(Q_1+Q_2)-2I_n+mB_3.
\end{equation*}
Let $B_3=-\frac{n}{m}(Q_1+Q_2)$.
Then $\rho_B$ is a scalar multiple of $I_n$.
Moreover, $\rho_{AB}$ is a self-adjoint matrix in $M_m(\mathbb{C})\otimes M_n(\mathbb{C})$.
We summarize the above argument as follows.

\begin{example}\label{eg tensor}
Suppose that $P_1,P_2$ are rank-one projections in $M_m(\mathbb{C})$ and $Q_1,Q_2$ are rank-one projections in $M_n(\mathbb{C})$.
Let
\begin{equation*}
	\rho_{AB}=P_1\otimes B_1+P_2\otimes B_2+I_m\otimes B_3\in M_m(\mathbb{C})\otimes M_n(\mathbb{C}),
\end{equation*}
where
\begin{equation*}
	B_1=(n-1)Q_1+Q_2-I_n,\quad B_2=Q_1+(n-1)Q_2-I_n,\quad B_3=-\frac{n}{m}(Q_1+Q_2).
\end{equation*}
Then $\rho_{AB}$ is self-adjoint with partial traces $\rho_A=-\frac{2n}{m}I_m$ and $\rho_B=-2I_n$.
\end{example}

We present a well-known result as follows.
See \cite[Proposition 11.1.8]{KR2} for a proof.

\begin{lemma}\label{lem tensor=0}
Suppose that $\{A_j\}_{j=1}^k\subseteq M_m(\mathbb{C})$, $\{B_j\}_{j=1}^k\subseteq M_n(\mathbb{C})$, and
\begin{equation*}
	\sum_{j=1}^kA_j\otimes B_j=0.
\end{equation*}
If the set $\{A_j\}_{j=1}^k$ is linearly independent, then $B_j=0$ for each $1\leqslant j\leqslant k$.
\end{lemma}

In the following technical lemma, we provide a sufficient condition to distinguish two matrices of the form given by \Cref{eg tensor}.

\begin{lemma}\label{lem 2-sigma}
Let $E_1,E_2,E'_1,E'_2$ be rank-one projections in $M_m(\mathbb{C})$, $F_1,F_2,F'_1,F'_2$ be rank-one projections in $M_n(\mathbb{C})$, and
\begin{equation*}
	\sigma_{AB}=E_1\otimes D_1+E_2\otimes D_2+I_m\otimes D_3,\quad
	\sigma'_{AB}=E'_1\otimes D'_1+E'_2\otimes D'_2+I_m\otimes D'_3,
\end{equation*}
where
\begin{equation*}
	\begin{cases}
		D_1=(n-1)F_1+F_2-I_n,\quad D_2=F_1+(n-1)F_2-I_n,\quad D_3=-(F_1+F_2),\\
		D'_1=(n-1)F'_1+F'_2-I_n,\quad D'_2=F'_1+(n-1)F'_2-I_n,\quad D'_3=-(F'_1+F'_2).
	\end{cases}
\end{equation*}
If $m\geqslant 3$, $n\geqslant 3$ and
\begin{equation*}
	E_1E_2=0,\quad F_1F_2\ne F_2F_1,\quad E'_1E'_2\ne E'_2E'_1,\quad F'_1F'_2=0,
\end{equation*}
then $\sigma_{AB}\ne\sigma'_{AB}$.
\end{lemma}

\begin{proof}
Suppose to the contrary that $\sigma_{AB}=\sigma'_{AB}$.
Then
\begin{equation*}
	E_1\otimes D_1+E_2\otimes D_2-E'_1\otimes D'_1-E'_2\otimes D'_2+I_m\otimes(D_3-D'_3)=0.
\end{equation*}
Let $\mathcal{A}=\mathrm{span}\{I_m,E_1,E_2,E'_1,E'_2\}$.
Since $E_1$ and $E_2$ are rank-one projections in $M_m(\mathbb{C})$ with $E_1E_2=0$ and $m\geqslant 3$, the set $\{I_m,E_1,E_2\}$ is linearly independent.
It follows that $3\leqslant\dim\mathcal{A}\leqslant 5$.
If $\dim\mathcal{A}=3$, then both $E'_1$ and $E'_2$ are linear combinations of $\{I_m,E_1,E_2\}$.
It follows that $E'_1E'_2=E'_2E'_1$, a contradiction.
If $\dim\mathcal{A}=4$, then without loss of generality, we may assume that $\mathcal{A}=\mathrm{span}\{I_m,E_1,E_2,E'_1\}$, i.e., we can write
\begin{equation*}
	E'_2=aI_m+bE_1+cE_2+dE'_1.
\end{equation*}
It follows that
\begin{align*}
	E_1\otimes(D_1-bD'_2)+E_2\otimes(D_2-cD'_2) & +E'_1\otimes(-D'_1-dD'_2)\\
	&+I_m\otimes(D_3-D'_3-aD'_2)=0.
\end{align*}
By \Cref{lem tensor=0}, we have $D_1=bD'_2$.
That is a contradiction because $D_1$ is invertible and $D'_2$ is not invertible by \Cref{lem P+Q-1}.
If $\dim\mathcal{A}=5$, then $D_1=(n-1)F_1+F_2-I_n=0$ by \Cref{lem tensor=0}, contradicting the assumption $n\geqslant 3$.
This completes the proof.
\end{proof}

The map $\Phi$ in the following lemma is called \emph{Sakai's flip} \cite{Sak-75}.

\begin{lemma}\label{lem flip}
Let $\Phi$ be a linear map on $M_n(\mathbb{C})\otimes M_n(\mathbb{C})$ such that $\Phi(A\otimes B)=B\otimes A$ for all $A,B\in M_n(\mathbb{C})$.
Then there exists a unitary matrix $W$ in $M_n(\mathbb{C})\otimes M_n(\mathbb{C})$ such that $\Phi(X)=WXW^*$ for all $X\in M_n(\mathbb{C})\otimes M_n(\mathbb{C})$.
\end{lemma}

\begin{proof}
It is clear that $\Phi$ is an automorphism of $M_n(\mathbb{C})\otimes M_n(\mathbb{C})\cong M_{n^2}(\mathbb{C})$.
Since every automorphism of $M_{n^2}(\mathbb{C})$ is inner, we can find such a unitary matrix $W$.
More precisely, let $\{e_j\}_{j=1}^n$ be an orthonormal basis for $\mathbb{C}^n$.
Then $W$ can be taken as the unitary matrix sending $e_i\otimes e_j$ to $e_j\otimes e_i$ for all $1\leqslant i,j\leqslant n$.
\end{proof}

Now we prove our main result in this paper.

\begin{theorem}\label{thm main}
For each $n\geqslant 3$, there are two-qudit states $\rho_{AB}$ and $\sigma_{AB}$ in $D(\mathbb{C}^n\otimes\mathbb{C}^n)$ and a unitary matrix $W$ in $M_n(\mathbb{C})\otimes M_n(\mathbb{C})$ such that
\begin{equation*}
	\begin{cases}
		\sigma_{AB} & =W\rho_{AB}W^*,\\
		\sigma_A\otimes I_n &=W(\rho_A\otimes I_n)W^*,\\
		I_n\otimes\sigma_B &=W(I_n\otimes\rho_B)W^*,
	\end{cases}
\end{equation*}
while $\sigma_{AB}$ and $\rho_{AB}$ are not locally unitarily equivalent.
\end{theorem}

\begin{proof}
Let $P_1,P_2,Q_1,Q_2$ be rank-one projections in $M_n(\mathbb{C})$ such that
\begin{equation*}
	P_1P_2=0,\quad Q_1Q_2\ne Q_2Q_1.
\end{equation*}
Let $\rho_{AB}=P_1\otimes B_1+P_2\otimes B_2+I_n\otimes B_3$, where
\begin{equation*}
	B_1=(n-1)Q_1+Q_2-I_n,\quad B_2=Q_1+(n-1)Q_2-I_n,\quad B_3=-(Q_1+Q_2).
\end{equation*}
Let $\sigma_{AB}=\Phi(\rho_{AB})$, where $\Phi$ is given by \Cref{lem flip}.
Then $\sigma_{AB}=W\rho_{AB}W^*$.
As in \Cref{eg tensor}, we see that
\begin{equation*}
	\sigma_A=\sigma_B=\rho_A=\rho_B=-2I_n.
\end{equation*}
Thus, it is clear that $\sigma_A\otimes I_n=W(\rho_A\otimes I_n)W^*$ and $I_n\otimes\sigma_B=W(I_n\otimes\rho_B)W^*$.
We claim that $\sigma_{AB}$ and $\rho_{AB}$ are not locally unitarily equivalent.
Suppose to the contrary that there are unitary matrices $U$ and $V$ in $M_n(\mathbb{C})$ such that
\begin{equation*}
	\sigma_{AB}=(U\otimes V)\rho_{AB}(U\otimes V)^*.
\end{equation*}
Let $E_j=UP_jU^*$ and $F_j=VQ_jV^*$ for $j=1,2$.
Then
\begin{equation*}
	\sigma_{AB}=E_1\otimes D_1+E_2\otimes D_2+I_n\otimes D_3,
\end{equation*}
where
\begin{equation*}
	D_1=(n-1)F_1+F_2-I_n,\quad D_2=F_1+(n-1)F_2-I_n,\quad D_3=-(F_1+F_2).
\end{equation*}
On the other hand, since $\sigma_{AB}=\Phi(\rho_{AB})$, we have
\begin{align*}
	\sigma_{AB}
	&=B_1\otimes P_1+B_2\otimes P_2+B_3\otimes I_n\\
	&=Q_1\otimes D'_1+Q_2\otimes D'_2+I_n\otimes D'_3,
\end{align*}
where
\begin{equation*}
	D'_1=(n-1)P_1+P_2-I_n,\quad D'_2=P_1+(n-1)P_2-I_n,\quad D'_3=-(P_1+P_2).
\end{equation*}
We get a contradiction by \Cref{lem 2-sigma}.

Note that the above $\rho_{AB}$ and $\sigma_{AB}$ are self-adjoint matrices in $M_n(\mathbb{C})\otimes M_n(\mathbb{C})$, but not two-qudit states in $D(\mathbb{C}^n\otimes\mathbb{C}^n)$.
Note that $\|\rho_{AB}\|\leqslant 2n+2$, $\mathrm{Tr}(\rho_{AB})=-2n$, and $\mathrm{Tr}(\rho_{AB}+(2n+2)I_n\otimes I_n)=2n(n^2+n-1)$.
Let
\begin{equation*}
	\rho'_{AB}=\frac{1}{2n(n^2+n-1)}(\rho_{AB}+(2n+2)I_n\otimes I_n).
\end{equation*}
Then $\rho'_{AB}$ is a maximally mixed marginal state in $D(\mathbb{C}^n\otimes\mathbb{C}^n)$.
Let $\sigma'_{AB}=\Phi(\rho'_{AB})$.
Then the pair $(\rho'_{AB},\sigma'_{AB})$ has the desired properties.
\end{proof}

\begin{remark}\label{rem main}
Suppose that $d=\mathrm{gcd}(m,n)\geqslant 3$.
If we write $m=dm_1$ and $n=dn_1$, then $\mathbb{C}^m=\mathbb{C}^d\otimes\mathbb{C}^{m_1}$ and $\mathbb{C}^n=\mathbb{C}^d\otimes\mathbb{C}^{n_1}$.
In this case, we can also find a pair $(\rho_{AB},\sigma_{AB})$ with the desired properties as in \Cref{thm main}.
\end{remark}

For the case $n\geqslant 4$, we provide a simplified example as follows.

\begin{example}
Assume that $n\geqslant 4$.
Let $P_1$ and $P_2$ be projections in $M_n(\mathbb{C})$ such that $\mathrm{Tr}(P_1)=1$ and $\mathrm{Tr}(P_2)=2$.
Let
\begin{equation*}
  \rho_{AB}=\frac{1}{2n}P_1\otimes P_2+\frac{1}{2n(n-1)}(I_n-P_1)\otimes(2I_n-P_2).
\end{equation*}
Then $\rho_{AB}$ is a two-qudit state in $D(\mathbb{C}^n\otimes\mathbb{C}^n)$ such that
\begin{equation*}
	\rho_A=\frac{1}{n}P_1+\frac{1}{n}(I_n-P_1)=\frac{1}{n}I_n,\quad
	\rho_B=\frac{1}{2n}P_2+\frac{1}{2n}(2I_n-P_2)=\frac{1}{n}I_n.
\end{equation*}
Similarly, let
\begin{equation*}
  \sigma_{AB}=\frac{1}{2n}P_2\otimes P_1+\frac{1}{2n(n-1)}(2I_n-P_2)\otimes(I_n-P_1).
\end{equation*}
Then $\sigma_A=\sigma_B=\frac{1}{n}I_n$.
By \Cref{lem flip}, $\sigma_{AB}$ and $\rho_{AB}$ are unitarily equivalent maximally mixed marginal states in $D(\mathbb{C}^n\otimes\mathbb{C}^n)$.

Suppose to the contrary that $\sigma_{AB}$ and $\rho_{AB}$ are locally unitarily equivalent.
Then there are unitary matrices $U_1$ and $U_2$ in $M_n(\mathbb{C})$ such that
\begin{equation*}
  \sigma_{AB}=(U_1\otimes U_2)\rho_{AB}(U_1\otimes U_2).
\end{equation*}
Let $Q_j=U_jP_jU_j^*$ for $j=1,2$.
Then
\begin{equation*}
	(n-1)P_2\otimes P_1+(2I_n-P_2)\otimes(I_n-P_1)
	=(n-1)Q_1\otimes Q_2+(I_n-Q_1)\otimes(2I_n-Q_2).
\end{equation*}
It follows that
\begin{equation*}
	P_2\otimes(nP_1-I_n)+Q_1\otimes(2I_n-nQ_2)+I_n\otimes(Q_2-2P_1)=0.
\end{equation*}
Since $\mathrm{rank}(Q_1)=1$, $\mathrm{rank}(P_2)=2$, and $n\geqslant 4$, the set $\{Q_1,P_2,I_n\}$ is linearly independent.
It follows that $nP_1-I_n=0$ by \Cref{lem tensor=0}.
That is a contradiction.
\end{example}

At the end of this paper, we show that for pure states in $D(\mathbb{C}^m\otimes\mathbb{C}^n)$, i.e., rank-one projections in $M_m(\mathbb{C})\otimes M_n(\mathbb{C})$, the local unitary equivalence is determined by the local unitary Bargmann invariants.

\begin{proposition}\label{prop rank-one}
Suppose that $\rho_{AB}$ and $\sigma_{AB}$ are two bipartite states in $D(\mathbb{C}^m\otimes\mathbb{C}^n)$ and $W$ is a unitary matrix in $M_m(\mathbb{C})\otimes M_n(\mathbb{C})$ such that
\begin{equation*}
	\begin{cases}
		\sigma_{AB} & =W\rho_{AB}W^*,\\
		\sigma_A\otimes I_n &=W(\rho_A\otimes I_n)W^*,\\
		I_m\otimes\sigma_B &=W(I_m\otimes\rho_B)W^*.
	\end{cases}
\end{equation*}
If $\rho_{AB}$ is a pure state, then $\sigma_{AB}$ and $\rho_{AB}$ are locally unitarily equivalent.
\end{proposition}

\begin{proof}
Suppose that $\rho_{AB}$ is the rank-one projection onto $\mathbb{C}\xi$, where $\xi$ is a unit vector in $\mathbb{C}^m\otimes\mathbb{C}^n$.
By the Schmidt decomposition, there are orthonormal sets $\{e_j\}_{j=1}^{k}\subseteq\mathbb{C}^m$, $\{f_j\}_{j=1}^{k}\subseteq\mathbb{C}^n$, and scalars $\lambda_1\geqslant\lambda_2\geqslant\cdots\geqslant\lambda_k>0$ with $1=\sum_{j=1}^{k}\lambda_j$ such that
\begin{equation*}
	\xi=\sum_{j=1}^{k}\sqrt{\lambda_j}e_j\otimes f_j.
\end{equation*}
For every $1\leqslant i,j\leqslant k$, we define rank-one matrices $E_{ij}\in M_m(\mathbb{C})$ and $F_{ij}\in M_n(\mathbb{C})$ by $E_{ij}x=\langle x,e_j\rangle e_i$ and $F_{ij}y=\langle y,f_j\rangle f_i$, respectively.
Then
\begin{equation*}
	\rho_{AB}=\sum_{i,j=1}^{k}\sqrt{\lambda_i\lambda_j}E_{ij}\otimes F_{ij},\quad
	\rho_A=\sum_{j=1}^{k}\lambda_jE_{jj},\quad
	\rho_B=\sum_{j=1}^{k}\lambda_jF_{jj}.
\end{equation*}
Since $\sigma_{AB}=W\rho_{AB}W^*$ is a rank-one projection, we can similarly write
\begin{equation*}
	\sigma_{AB}=\sum_{i,j=1}^{k'}\sqrt{\lambda'_i\lambda'_j}E'_{ij}\otimes F'_{ij},\quad
	\sigma_A=\sum_{j=1}^{k'}\lambda'_jE'_{jj},\quad
	\sigma_B=\sum_{j=1}^{k'}\lambda'_jF'_{jj}.
\end{equation*}
Since $\sigma_A\otimes I_n=W(\rho_A\otimes I_n)W^*$, we see that $k'=k$ and $\lambda'_j=\lambda_j$ for every $1\leqslant j\leqslant k$.
From this, it is routine to see that $\sigma_{AB}$ and $\rho_{AB}$ are locally unitarily equivalent.
\end{proof}



\end{document}